\documentclass[american,aps,pra,reprint, superscriptaddress]{revtex4-1}
\usepackage[unicode=true,pdfusetitle, bookmarks=true,bookmarksnumbered=false,bookmarksopen=false, breaklinks=false,pdfborder={0 0 0},backref=false,colorlinks=false] {hyperref}
\hypersetup{ colorlinks,linkcolor=myurlcolor,citecolor=myurlcolor,urlcolor=myurlcolor}
\usepackage{graphics,epstopdf,graphicx, amsthm, amsmath, amssymb, times, braket, color, bm}
\usepackage[up]{subfigure}
\usepackage{cleveref}

\definecolor{myurlcolor}{rgb}{0,0,0.7}
\newcommand{\proj}[1]{| #1\rangle\!\langle #1 |}
\newcommand{\Br}[1]{\left[#1\right]}

\def\md{{(\mathrm{d})}}
\theoremstyle{plain}
\newtheorem{thm}{\protect\theoremname}
\newtheorem{prop}[thm]{Proposition}
\newtheorem{lem}[thm]{Lemma}
\providecommand{\theoremname}{Theorem}
\newcommand*{\myproofname}{Proof}
\newenvironment{mproof}[1][\myproofname]{\begin{proof}[#1]}{\end{proof}}
\makeatother

\begin{document}

 \author{Kaifeng Bu}
 \email{bkf@zju.edu.cn}
 \affiliation{School of Mathematical Sciences, Zhejiang University, Hangzhou 310027, PR~China}

 \author{Uttam Singh}
 \email{uttamsingh@hri.res.in}
 \affiliation{Harish-Chandra Research Institute, Allahabad, 211019, India}

 \author{Junde Wu}
 \email{wjd@zju.edu.cn}
 \affiliation{School of Mathematical Sciences, Zhejiang University, Hangzhou 310027, PR~China}

\title{Catalytic coherence transformations}
\begin{abstract}
Catalytic coherence transformations allow the otherwise impossible state transformations using only incoherent operations with the aid of an auxiliary system with finite coherence which is not being consumed in anyway. Here we find the necessary and sufficient conditions for the deterministic and stochastic catalytic coherence transformations between pair of pure quantum states. In particular, we show that the simultaneous decrease of a family of R\'enyi entropies of the diagonal parts of the states under consideration are necessary and sufficient conditions for the deterministic catalytic coherence transformations. Similarly, for stochastic catalytic coherence transformations we find the necessary and sufficient conditions for achieving higher optimal probability of conversion. We, thus, completely characterize the coherence transformations amongst pure quantum states under incoherent operations. We give numerous examples to elaborate our results. We also explore the possibility of the same system acting as a catalyst for itself and find that indeed {\it self catalysis} is possible. Further, for the cases where no catalytic coherence transformation is possible we provide entanglement assisted coherence transformations and find the necessary and sufficient conditions for such transformations.
\end{abstract}
\maketitle

\section{Introduction}
Quantum resource theories \cite{Oppenheim13, FBrandao15} have been a corner stone to the development and quantitative understanding of various physical phenomena in quantum physics and quantum information theory. A resource theory comprises of two basic elements: one is the set of allowed (free) operations and other being the set of allowed (free) states. Any operation (or state) falling out of the set of free operations (or the set of free states) is then dubbed as a resource. The most prominent resource theory is the resource theory of entanglement \cite{HorodeckiRMP09}. The other notable examples include the resource theories of thermodynamics \cite{Fernando2013}, asymmetry \cite{Gour2008}, coherence \cite{Baumgratz2014, Marvian14} and steering \cite{Rodrigo2015}. The main advantages of having a resource theory for some phenomenon are the succinct understanding of various physical processes and operational quantification of the relevant resources. Therefore, a major concern of any resource theory is to describe and uncover the intricate structure of the physical processes (state transformations) within the set of allowed operations. The possibility of catalysis is one such phenomenon which allows the otherwise impossible state transformations via the set of allowed operations in a given resource theory. This is very natural as the additional systems (catalysts) are always available and importantly, in such transformations the additional resources are not consumed in anyway. The catalysis in quantum resource theories was first introduced in Refs. \cite{Jonathan1999, JonathanB1999} in the context of entanglement. The consideration of catalysts in the resource theory of thermodynamics turned out to be very surprising and extremely important that has led to the introduction of many second laws of quantum thermodynamics \cite{Brandao2015} compared to the single second law in the macroscopic thermodynamics \cite{Callen1985}. Recently, in the context of entanglement it is found that a quantum system can act as a catalyst for itself, allowing further the possibility of {\it self catalysis} \cite{Duarte2015}.

\begin{figure}
\label{fig1}
\includegraphics[width=50mm]{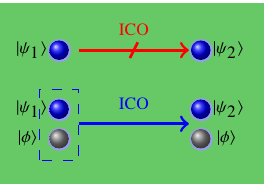}
\caption{The schematic for catalytic coherence transformations. Consider a finite dimensional quantum system in state $\ket{\psi_1}$. Let $\ket{\psi_2}$ be an incomparable state to $\ket{\psi_1}$, i.e., using only incoherent operations one cannot convert $\ket{\psi_1}$ into $\ket{\psi_2}$ with certainty. However, if one has temporary access to another coherent state $\ket{\phi}$, one can always achieve the transformation from $\ket{\psi_1}$ to $\ket{\psi_2}$. The state $\ket{\phi}$ is not consumed in any way and can, therefore, be viewed as a catalyst for this transformation.}
\label{fig}
\end{figure}

Catalysis in the resource theory of coherence was first considered in Ref. \cite{Aberg14} and is developed since then (see Refs. \cite{Du2015, DuB2015}, also see Fig. \ref{fig1}). In this work we further delineate the phenomenon of catalysis in the resource theory of coherence both in the deterministic and stochastic scenarios and completely characterize the coherence transformations amongst pure quantum states under incoherent operations. This is an important step towards a complete theory of quantum coherence based on incoherent operations as the allowed operations \cite{Baumgratz2014}. In particular, we obtain the necessary and sufficient conditions for the deterministic and stochastic coherence transformations in the presence of catalysts. We first find the necessary and sufficient conditions for the enhancement of the optimal probability of conversion while the catalysts are available. Then we go on to find the necessary and sufficient conditions for the deterministic catalytic coherence transformations and show that these are given by the simultaneous decrease of a family of R\'enyi entropies of the diagonal parts of the states under consideration in a fixed basis. This result is very similar in nature to the many second laws of quantum thermodynamics. We also provide a dedicated discussion on the practicality of these necessary and sufficient conditions. Further, for the cases where no catalytic coherence transformation is possible we consider the possibility of entanglement assisted coherence transformations and find the necessary and sufficient conditions for the same. Furthermore, we find that {\it self catalysis} in the context of coherence resource theory is possible for the transformations of certain states. We hope that our results will be useful for coherence transformations in the resource theory of coherence, in situations where processing of coherence is limited by additional restrictions from quantum thermodynamics and in context of {\it single-shot} information theory \cite{TomamichelPhD2012}.

The paper is organized as follows. We start with a discussion on interconversion of quantum states under incoherent operations, measures of coherence and coherence transformations for pure states along with some other preliminaries in Sec. \ref{prelims}. In Sec. \ref{sec:cat-coh-trans}, we discuss and obtain various results on catalytic coherence transformations under deterministic and stochastic scenarios. We present the necessary and sufficient conditions for catalytic coherence transformations in Sec. \ref{sec:nec-suff}. We then go beyond catalytic coherence transformations to entanglement assisted incoherent transformations in Sec. \ref{sec:b-cat-coh}. Finally, we conclude in Sec. \ref{sec:conclusion} with overview and implications of the results presented in the paper. The appendix lists some useful results that are obtained earlier by other researchers.

\section{Preliminaries}
\label{prelims}
\emph{Interconversion of quantum states, the set of allowed operations and measures of coherence.--} The notion of interconvertibility of quantum states is desirable in many situations. For example, if we have a reference quantum state that serves as a basic unit of certain resource such as coherence or entanglement (just like Kilogram serves as a basic unit for mass), then one would like to convert any other given state to the reference state and in this way one can estimate the amount of resource in a given state. The distillable entanglement \cite{Bennett1996, BennettB1996, Rains1999} of a bipartite quantum state is one of the measures of entanglement that is obtained via converting the state into the maximally entangled state. The other examples include the distillable coherence and coherence of formation \cite{Winter2015} (for single quantum systems), and entanglement of formation \cite{Bennettc1996, Wootters1998} (for bipartite quantum systems). We would like to emphasize here that the conversion from one state to the other is achieved by employing the relevant set of allowed operations. Currently, no common agreement persists for the definition of quantum coherence and we have two independent resource theories of coherence. One is based on the resource theory of asymmetry \cite{Gour2008, Marvian14} and turned out to be very successful in the thermodynamical contexts \cite{Rudolph214, Rudolph114}. The other one is based on the set of incoherent operations as the allowed operations \cite{Baumgratz2014}. We consider the latter resource theory of coherence throughout this work. In this resource theory of coherence, the set of incoherent operations is the allowed set of operations and any interconversion among quantum states is effected via operations from this set only. In quantum theory, a physically admissible operation $\Phi$ is a linear completely positive and trace-preserving map. Such a map $\Phi$ can be expressed by a set of Kraus operators $\{K_n\}_{n=1}^{N}$ such that $\Phi(\rho)=\sum_{n=1}^{N}K_n\rho K^\dag_n$, with $\sum_{n=1}^{N}K^\dag_nK_n=\mathbb{I}$ \cite{Nielsen10}. However, as mentioned, in the resource theory of coherence, the allowed operations are only the \emph{incoherent operations}.  An operation $\Phi_I$ is called an incoherent operation if the Kraus operators $\{K_n\}$ of $\Phi_I$ are such that $K_n\mathcal{I} K^{\dag}_n\subseteq \mathcal{I}$. Here $\mathcal{I}$ is the set of all incoherent states. Given a fixed reference basis, say $\{\ket{i}\}$, any state which is diagonal in the reference basis is called an incoherent state. It is to be noted that the notion of quantum coherence is basis dependent. The quantifiers of coherence in the resource theory of coherence based on the set of allowed operations as the incoherent operations have been shown to be operationally meaningful \cite{Alex15, UttamA2015, Winter2015} and hence establish the importance of this resource theory. The bona fide quantifiers of coherence include the $l_1$ norm of coherence, relative entropy of coherence \cite{Baumgratz2014} and R\'enyi entropies for certain range of R\'enyi index \cite{DuB2015}. For a pure state $\ket{\psi}$, the relative entropy of coherence $C_r(\ket{\psi})$ becomes the von Neumann entropy of its diagonal part in the fixed reference basis, i.e., $C_r(\ket{\psi}) = S (\psi^\md)$, where $S$ is the von Neumann entropy and $\psi^\md$ is the diagonal part of the state $\ket{\psi}$ in a fixed reference basis.

\emph{Catalysis.--} Just like the concept of catalysis in chemical reactions (conversion of a mixture of compounds into mixture of other compounds with the aid of a catalyst), there exists a similar concept in the context of interconversion of quantum states. Let us consider that we need a conversion of an initial state $\ket{\psi_1}$ into a final state $\ket{\psi_2}$ of a quantum system $\mathcal{H}$ by using only the restricted class of operations and assume further that this conversion is not possible. Now, if there exists a pure state $\ket{\phi}$ of the same system $\mathcal{H}$ or any other ancillary system $\mathcal{K}$ such that $\ket{\psi_1}\otimes\ket{\phi}$ can be transformed into $\ket{\psi_2}\otimes\ket{\phi}$ by using only the  restricted class of operations, then  such a transformation is called a catalytic transformation and $\ket{\phi}$ is called as a catalyst for the transformation $\ket{\psi_1}\rightarrow\ket{\psi_2}$. The state $\ket{\phi}$, just like a catalyst in a chemical process, does not change after the transformation (also see Fig. \ref{fig}). It is also possible that $n$ copies of same initial state $\ket{\psi_1}$ can act as a catalyst for the transformation $\ket{\psi_1}\rightarrow\ket{\psi_2}$, i.e., despite the impossibility of the transformation $\ket{\psi_1}\rightarrow\ket{\psi_2}$, the transformation $\ket{\psi_1}\otimes \ket{\psi_1}^{\otimes n}\rightarrow\ket{\psi_2}\otimes \ket{\psi_1}^{\otimes n}$ may be possible. Here, $n$ is a positive integer and depends on the transformation under consideration. This kind of catalysis is dubbed as {\it self catalysis} and we elaborate on it further as we go along. For the resource theory of coherence, we take incoherent operations for the restricted class of operations in the above definition, and the transformations then are referred to as catalytic coherence transformations.

\emph{Deterministic coherence transformations.--}
The possibility of transformation of a quantum system from one state with finite coherence to another state is determined by the majorization of the diagonal elements of the corresponding pure states in a fixed basis. This result was first proved in Ref. \cite{Du2015}. We state this result again for brevity:
\begin{thm}[\cite{Du2015}]
\label{th:DBG}
Let $\ket{\psi_1}$ and $\ket{\psi_2}$ be two pure states with $\psi^\md_1$ and $\psi^\md_2$ being the diagonal parts of $\ket{\psi_1}$ and $\ket{\psi_2}$, respectively in a fixed reference basis. Then a transformation from the state $\ket{\psi_1}$ to $\ket{\psi_2}$ is possible via incoherent operations {\it if and only if} $\psi^\md_1$ is majorized by $\psi^\md_2$, i.e., $\psi^\md_1 \prec\psi^\md_2$.
\end{thm}
For two probability vectors $p=\{p_i\}$ and $q=\{q_i\}$ $(i=1,\ldots,d)$ arranged in decreasing order, $p$ is said to be majorized by $q$, i.e., $p\prec q$ if $\sum_{i=1}^{l}p_i\leq \sum_{i=1}^{l}q_i$ for $l=1,\ldots,d-1$ and $\sum_{i=1}^{d}p_i=1= \sum_{i=1}^{d}q_i$.
%
Theorem \ref{th:DBG} is the key ingredient for discussing the incoherent transformations between two pure states and the problem at the hand. However, we find that the proof of the converse part of the Theorem \ref{th:DBG}, given in Ref. \cite{Du2015}, is true \emph{only} for a specific class of incoherent operations.  In the original proof, it was claimed that if a pure state $\ket{\psi}$ can be transformed to another pure state $\ket{\phi}$ through an incoherent channel then the elements $K_n$ of the channel can always be written as \cite{Du2015}
\begin{equation}
\label{eq:Kn}
K_n=P_{\pi_n} \Br{\begin{array}{ccc}
a^{(n)}_1 & \delta_{1,i(2)}a^{(n)}_2 & \delta_{1,i(3)} a^{(n)}_3  \\
0 & \delta_{2,i(2)}a^{(n)}_2 & \delta_{1, i(3)} a^{(n)}_3\\
0 & 0 & \delta_{3,i(3)}a^{(n)}_3
\end{array}},
\end{equation}
where $\delta_{ij}$ is the Kronecker delta function, $a^{(n)}_j$ $(j=1, 2, 3)$ is the nonzero entry of $K_n$ in the $j^{\mathrm{th}}$ column, $i(j)$ is the location of $a^{(n)}_j$ in the $i^{\mathrm{th}}$ row and was treated independent of $n$, and $P_{\pi_n}$ is the permutation matrix. This means that the channel elements $K_n$ considered in Ref. \cite{Du2015} were of the same kind (upper triangular matrices) up to permutations. However, this is not the case always. For example let the initial state be $\ket{\psi}=\sum^{2}_{i=0}\sqrt{1/3}\ket{i}$, the final state be $\ket{\phi}=\ket{0}$, and define an incoherent operation $\Phi=\{K_n\}_{n=1}^{8}$ with the Kraus elements $K_n$ being given by
\begin{align}
&K_1=\Br{\begin{array}{ccc}
\frac{1}{2\sqrt{2}} & \frac{1}{2\sqrt{2}} & \frac{1}{2\sqrt{2}}  \\
0 & 0 & 0\\
0 & 0 & 0
\end{array}}=K_2,\nonumber\\
&K_3=\Br{\begin{array}{ccc}
-\frac{1}{2\sqrt{2}} & \frac{1}{2\sqrt{2}} & \frac{1}{2\sqrt{2}}  \\
0 & 0 & 0\\
0 & 0 & 0
\end{array}}=K_4,\nonumber\\
&K_5=\Br{\begin{array}{ccc}
\frac{1}{2\sqrt{2}} & 0 & 0  \\
0 & \frac{1}{2\sqrt{2}} & -\frac{1}{2\sqrt{2}}\\
0 & 0 & 0
\end{array}}=K_6=K_7=K_8.\nonumber
\end{align}
It is easy to see $K_n\ket{\psi}=\alpha_n\ket{\phi}$ (for some $\alpha_n$ such that $\sum_n|\alpha_n|^2=1$), $\sum^8_{n=1}K^{\dag}_nK_n=I$ and $\Phi(\proj{\psi})=\sum^8_{n=1}K_n\proj{\psi}K^{\dag}_n=\proj{\phi}$. But, obviously, $K_1$ and $K_5$ are not related to each other via permutations and therefore, these are different kinds of upper triangular matrices. This means that $i(j)$ that was considered independent of $n$ must depend on $n$, in general.
This discrepancy has already been noticed in Ref. \cite{DuB2015} and amended by considering $i(j)$ that explicitly depend on $n$.

It is important to note that for coherence transformations of pure states via incoherent operations, without loss of generality, we can always assume that the coefficients of the pure states in a fixed reference basis are all real, positive and arranged in the decreasing order \cite{Du2015}. Throughout the paper we take this for granted and mention this at the places where we think it is necessary.

\emph{Stochastic coherence transformations.--} We discussed above the necessary and sufficient conditions for the successful transformation of an initial state $\ket{\psi_1}$ into a final state $\ket{\psi_2}$ by incoherent operations under the name of deterministic coherence transformation (as the probability to achieve the transformation is $1$).
For dimensions strictly greater than two, because the majorization is only a partial order, Theorem \ref{th:DBG} leaves us with the possibility that there can be a pair of states $\ket{\psi_1}$ and $\ket{\psi_2}$ such that neither $\psi^\md_1 \prec\psi^\md_2$ nor $\psi^\md_2 \prec\psi^\md_1$. These states will be called incomparable states in terms of coherence properties.
To deal with these incomparable states, the stochastic transformations have been investigated in Refs. \cite{Vidal1999, Feng2004, Feng2005, DuB2015}. Here, a stochastic coherence transformation means that for an incoherent operation $\Phi_I$ with Kraus operators $\set{K_n}_{n=1}^{N}$, although $\Phi_I$ cannot transform $\ket{\psi_1}$ into $\ket{\psi_2}$, i.e., $\ket{\psi_2}\bra{\psi_2}\neq\Phi_I[\ket{\psi_1}\bra{\psi_1}]$, there may exist a subset $\set{K_i}_{i=1}^{N'}$ of set $\set{K_n}_{n=1}^{N}$ ($N'<N$) such that $K_i\ket{\psi_1}\propto \ket{\psi_2}$. The maximum value of probability of transforming the initial state into the finial state, i.e., $\bra{\psi_1}\left(\sum_{i=1}^{N'}K_i^\dag K_i\right)\ket{\psi_1}$, has been calculated in \cite{DuB2015}. Here $\sum_{i=1}^{N'}K_i^\dag K_i\neq \mathbb{I}$. Thus, deterministic coherence transformations are the stochastic coherence transformations with optimal probability of transformation being equal to $1$.

\section{Deterministic and stochastic catalytic coherence transformations}
\label{sec:cat-coh-trans}
\emph{Catalysis under deterministic incoherent operations.--}
We know that there exists pair of incomparable quantum states such that any one of them cannot be transformed to another only using incoherent operations. Such examples can be constructed very easily. Let us consider a qutrit system with the states $\ket{\psi_1} = \sum_{i=0}^{2} \sqrt{\psi_1^i} \ket{i}$ and $\ket{\psi_2} = \sum_{i=0}^{2} \sqrt{\psi_2^i} \ket{i}$. Choose $\psi_1^i$ and $\psi_2^i$ such that $|\psi_1^0|\leq |\psi_2^0|$ and $|\psi_1^0| + |\psi_1^1| > |\psi_2^0| + |\psi_2^1|$. The diagonal parts of such states will never be majorized by one another. The specific examples are given in Table \ref{tab:ex}. Let us consider $d$-dimensional incomparable states $\ket{\psi_1} = \sum_{i=0}^{d-1} \sqrt{\psi_1^i} \ket{i}$ and $\ket{\psi_2} = \sum_{i=0}^{d-1} \sqrt{\psi_2^i} \ket{i}$. Despite the impossibility of transformation from $\ket{\psi_1}$ to $\ket{\psi_2}$ via incoherent operations, it is known that another auxiliary system with coherence (catalyst) can be used to make this transformation possible \cite{Du2015, DuB2015}  (via catalytic coherence transformations (see Fig. \ref{fig})).
There are following general properties of catalytic coherence transformations \cite{Du2015}: (a) No incoherent transformation can be catalyzed by a maximally coherent state $\ket{\psi_M} = \frac{1}{\sqrt{d}}\sum_{i=0}^{d-1}\ket{i}$. (b) Two states are interconvertible, i.e., $\ket{\psi_1}\rightleftharpoons\ket{\psi_2}$, under catalytic coherence transformations {\it if and only if} they are equivalent up to a permutation of diagonal unitary transformations. (c) $\ket{\psi_1}\rightarrow \ket{\psi_2}$ under catalytic coherence transformations only if both $|\psi_1^0|\leq |\psi_2^0|$ and $|\psi_1^{d-1}|\geq |\psi_2^{d-1}|$ hold. However, the necessary and sufficient conditions are hitherto missing for catalytic coherence transformations. We provide these conditions for stochastic coherence transformations later in this section and for the deterministic coherence transformations for pure quantum states in the next section.

But before going any further, let us consider a specific example of a qutrit system in the state $\ket{\psi_1} = \sqrt{0.4}\ket{0}+\sqrt{0.4}\ket{1}+\sqrt{0.1}\ket{2}+\sqrt{0.1}\ket{3}$. We want to make the otherwise impossible transformation from $\ket{\psi_1}$ to $\ket{\psi_2} = \sqrt{0.5}\ket{0}+\sqrt{0.25}\ket{1}+\sqrt{0.25}\ket{2}$ via incoherent operations using a catalyst in state $\ket{\phi}$. It can be seen that we can choose $\ket{\phi} = \sqrt{0.6}\ket{0} + \sqrt{0.4}\ket{1}$. In this case, we have $\ket{\psi_1}\otimes \ket{\phi}\xrightarrow{\mathrm{ICO}}\ket{\psi_2}\otimes \ket{\phi}$. Here ICO denotes the incoherent operation. It is important to note that the state $\ket{\phi}$ is not unique. For example in the above case $\ket{\phi} = \sqrt{0.62}\ket{0} + \sqrt{0.38}\ket{1}$ can also act as a catalyst. So it is a legitimate question to ask that what is the structure of the set of catalysts for a given catalytic transformation, i.e., for fixed $\ket{\psi_1}$ and $\ket{\psi_2}$, what is the set $\{\ket{\phi}\}$ such that $\ket{\psi_1}\otimes \ket{\phi}\xrightarrow{\mathrm{ICO}}\ket{\psi_2}\otimes \ket{\phi}$? The following proposition answers this question for four dimensional systems.
\begin{prop}\label{prop:str_of_cat}
Consider a four dimensional system with states $\ket{\psi_1}=\sum^4_{i=1}\sqrt{\psi_1^i}\ket{i}$ and $\ket{\psi_2}=\sum^4_{i=1}\sqrt{\psi_2^i}\ket{i}$ such that $\ket{\psi_1}\nrightarrow \ket{\psi_2}$ under incoherent operations. Without loss of generality we can assume that the coefficients $\{\psi_1^i\}$, $\{\psi_2^i\}$ are real and arranged in decreasing order. The necessary and sufficient conditions for the existence of a catalyst $\ket{\phi}=\sqrt{a}\ket{1}+\sqrt{1-a}\ket{2}$ $(a\in(0.5,1))$ for these states are: $\psi_1^1\leq \psi_2^1;~\psi_1^1+\psi_1^2>\psi_2^1+\psi_2^2;~\psi_1^1+\psi_1^2+\psi_1^3\leq\psi_2^1+\psi_2^2+\psi_2^3$, and
\begin{align}
&\max\left\{\frac{\psi_1^1+\psi_1^2-\psi_2^1}{\psi_2^2+\psi_2^3},1-\frac{\psi_1^4-\psi_2^4}{\psi_2^3-\psi_1^3}\right\}\nonumber\\
&\leq a\leq \min\left\{\frac{\psi_2^1}{\psi_1^1+\psi_1^2},\frac{\psi_2^1-\psi_1^1}{\psi_1^2-\psi_2^2},1-\frac{\psi_2^4}{\psi_1^3+\psi_1^4}\right\}.
\end{align}
\end{prop}
\begin{mproof}
For $\ket{\psi_1}=\sum^4_{i=1}\sqrt{\psi_1^i}\ket{i}$, $\ket{\psi_2}=\sum^4_{i=1}\sqrt{\psi_2^i}\ket{i}$ and $\ket{\phi}=\sqrt{a}\ket{1}+\sqrt{1-a}\ket{2}$, we can define $\ket{\gamma_1}_{AB}=\sum^4_{i=1}\sqrt{\psi_1^i}\ket{i}\ket{i}$, $\ket{\gamma_2}_{AB}=\sum^4_{i=1}\sqrt{\psi_2^i}\ket{i}\ket{i}$ and $\ket{\eta}_{AB}=\sqrt{a}\ket{11}+\sqrt{1-a}\ket{22}$. Then $\psi^{(d)}_1\otimes \phi^{(d)}\prec\psi^{(d)}_2\otimes \phi^{(d)}$ is equivalent to $\mathrm{Tr}_A(\proj{\gamma_1}\otimes \proj{\eta})\prec \mathrm{Tr}_A(\proj{\gamma_2}\otimes \proj{\eta})$. Now the proof of our proposition follows from the Theorem \ref{th:Sun2005} of appendix which was proved in Ref. \cite{Sun2005}.
\end{mproof}
It may be noted that based on the connections between the resource theories of coherence and entanglement, the results of the catalytic transformations in entanglement theory can always  be carried over to the coherence theory.
\begin{table}
\centering
\caption{\label{tab:ex} Some examples of incomparable (via incoherent operations) coherent states in the computational basis.}
\begin{tabular}{|l|l|}
\hline
$\ket{\psi_1}$ & $\ket{\psi_2}$ \\ \hline
$\sqrt{0.5}\ket{0}+\sqrt{0.4}\ket{1}+\sqrt{0.1}\ket{2}$ & $\sqrt{0.6}\ket{0}+\sqrt{0.2}\ket{1}+\sqrt{0.2}\ket{2}$ \\ \hline
$\sqrt{0.5}\ket{0}+\sqrt{0.4}\ket{1}+\sqrt{0.1}\ket{2}$ & $\sqrt{0.6}\ket{0}+\sqrt{0.25}\ket{1}+\sqrt{0.15}\ket{2}$ \\ \hline
$\sqrt{0.5}\ket{0}+\sqrt{0.4}\ket{1}+\sqrt{0.1}\ket{2}$ & $\sqrt{0.7}\ket{0}+\sqrt{0.15}\ket{1}+\sqrt{0.15}\ket{2}$ \\ \hline
$\sqrt{0.4}\ket{0}+\sqrt{0.4}\ket{1}+\sqrt{0.2}\ket{2}$ & $\sqrt{0.45}\ket{0}+\sqrt{0.3}\ket{1}+\sqrt{0.25}\ket{2}$ \\ \hline
$\sqrt{0.4}\ket{0}+\sqrt{0.4}\ket{1}+\sqrt{0.2}\ket{2}$ & $\sqrt{0.5}\ket{0}+\sqrt{0.25}\ket{1}+\sqrt{0.25}\ket{2}$ \\ \hline
\end{tabular}
\end{table}
Also, it is noted that if for the states $\ket{\psi_1} = \sum_{i=0}^{d-1} \sqrt{\psi_1^i} \ket{i}$ and $\ket{\psi_2} = \sum_{i=0}^{d-1} \sqrt{\psi_2^i} \ket{i}$, $\ket{\phi}$ is a catalyst then $\ket{\phi}$ acts as a catalyst for the states $\ket{\psi_1} = \sum_{i=0}^{d} \sqrt{\tilde{\psi}_1^i} \ket{i}$ and $\ket{\psi_2} = \sum_{i=0}^{d} \sqrt{\tilde{\psi}_2^i} \ket{i}$, where $\tilde{\psi}_k^i = \psi_k^i$ for $k=1,2$ and $i=1,\cdots, d-2$. $\tilde{\psi}_k^{d-1} = \psi_k^{d-1}-\epsilon_k$ and $\tilde{\psi}_k^{d} =\epsilon_k$ for $k=1,2$. For example, for the states $\ket{\psi_1} = \sqrt{0.4}\ket{0}+\sqrt{0.4}\ket{1}+\sqrt{0.1}\ket{2}+\sqrt{0.1}\ket{3}$ and $\ket{\psi_2} = \sqrt{0.5}\ket{0}+\sqrt{0.25}\ket{1}+\sqrt{0.25}\ket{2}$ the catalyst is $\ket{\phi} = \sqrt{0.6}\ket{0} + \sqrt{0.4}\ket{1}$. Now for the states $\ket{\psi_1} = \sqrt{0.4}\ket{0}+\sqrt{0.4}\ket{1}+\sqrt{0.1}\ket{2}+\sqrt{0.05}\ket{3}+\sqrt{0.05}\ket{4}$ and  $\ket{\psi_2} = \sqrt{0.5}\ket{0}+\sqrt{0.25}\ket{1}+\sqrt{0.25}\ket{2}$ the catalyst can again be chosen as $\ket{\phi} = \sqrt{0.6}\ket{0} + \sqrt{0.4}\ket{1}$. Moreover, if $\ket{\psi_1}\nrightarrow\ket{\psi_2}$, it is possible that $\ket{\psi_1}\otimes \ket{\psi_1}^{\otimes N}\rightarrow\ket{\psi_2}\otimes \ket{\psi_1}^{\otimes N}$ for some $N\geq1$. This means that a state can act as catalyst for itself. For example, take states $\ket{\psi_1} = \sqrt{0.9}\ket{0}+\sqrt{0.081}\ket{1}+\sqrt{0.01}\ket{2}+\sqrt{0.009}\ket{3}$ and  $\ket{\psi_2} = \sqrt{0.95}\ket{0}+\sqrt{0.03}\ket{1}+\sqrt{0.02}\ket{2}$. $\ket{\psi_1}$ acts as a catalyst here, i.e., $\ket{\psi_1}\nrightarrow\ket{\psi_2}$ but $\ket{\psi_1}\otimes \ket{\psi_1}\rightarrow\ket{\psi_2}\otimes \ket{\psi_1}$. Similarly, if we take $\ket{\psi_1} = \sqrt{0.9}\ket{0}+\sqrt{0.088}\ket{1}+\sqrt{0.006}\ket{2}+\sqrt{0.006}\ket{3}$ and $\ket{\psi_2} = \sqrt{0.95}\ket{0}+\sqrt{0.03}\ket{1}+\sqrt{0.02}\ket{2}$ the two copies of $\ket{\psi_1}$ act as a catalyst, i.e., $\ket{\psi_1}\nrightarrow\ket{\psi_2}$ but $\ket{\psi_1}\otimes \ket{\psi_1}^{\otimes 2}\rightarrow\ket{\psi_2}\otimes \ket{\psi_1}^{\otimes 2}$. These examples are taken from Ref. \cite{Duarte2015} which deals with {\it self catalysis} in entanglement theory.

\smallskip
\noindent
\emph{Catalysis under stochastic incoherent operations.--} Here we explore the possibility of transforming a pure state to another incomparable pure state using stochastic incoherent operations as we already know that for such a pair of states there doest not exist any deterministic incoherent operation that can facilitate this transformation. We consider the transformations both in presence and in absence of catalysts. It is obtained in Ref. \cite{DuB2015} that for a pure state $\ket{\psi}=\sum_{i=0}^{d-1}\sqrt{\psi_i}\ket{i}$, the functions $C_l(\psi) = \sum_{i=l}^{d-1}|\psi_i|$, $l=0,\cdots,d-1$ are valid coherence measures in the sense of the resource theory of coherence \cite{Baumgratz2014}. Moreover, in the absence of catalysts, the optimal probability $P\left(\ket{\psi_1}\xrightarrow{\mathrm{ICO}}\ket{\psi_2}\right)$ of converting a pure state $\ket{\psi_1}$ into $\ket{\psi_2}$ is given by \cite{DuB2015}
\begin{align}
\label{eq:opt-prob}
 P\left(\ket{\psi_1}\xrightarrow{\mathrm{ICO}}\ket{\psi_2}\right) = \min_{l\in[0,d-1]}\frac{C_l(\psi_1)}{C_l(\psi_2)}.
\end{align}

We first prove that the optimal probability of deterministic incoherent state transformations is always one. Consider a pair of pure states $\ket{\psi_1} = \sum_{i=0}^{d-1} \sqrt{\psi_1^i} \ket{i}$ and $\ket{\psi_2} = \sum_{i=0}^{d-1} \sqrt{\psi_2^i} \ket{i}$ such that $\ket{\psi_1}\xrightarrow{\mathrm{ICO}}\ket{\psi_2}$. From Theorem \ref{th:DBG}, if $\ket{\psi_1}\xrightarrow{\mathrm{ICO}}\ket{\psi_2}$, then $(\psi^0_1,\ldots,\psi^{d-1}_1)\prec(\psi^0_2,\ldots,\psi^{d-1}_2)$. Thus
$\sum^l_{i=0}\psi^i_1\leq\sum^l_{i=0}\psi^i_2$. Due to normalization, we have
\begin{eqnarray*}
\sum^{d-1}_{i=l}\psi^i_1\geq\sum^{d-1}_{i=l}\psi^i_2.
\end{eqnarray*}
That is $C_l(\psi_1)\geq C_l(\psi_2)$ for all $l$ values. Hence, $P(\ket{\psi_1}\xrightarrow{\mathrm{ICO}}\ket{\psi_2})=1$.
We note that mere the presence of another quantum system (a catalyst) can enhance the optimal probability of transition given by Eq. \eqref{eq:opt-prob}. For example consider the states $\ket{\psi_1} = \sqrt{0.4}\ket{0}+\sqrt{0.4}\ket{1}+\sqrt{0.2}\ket{2}$ and $\ket{\psi_2}= \sqrt{0.5}\ket{0}+\sqrt{0.25}\ket{1}+\sqrt{0.25}\ket{2}$. Here $P\left(\ket{\psi_1}\xrightarrow{\mathrm{ICO}}\ket{\psi_2}\right) = 0.8$ and $P\left(\ket{\psi_2}\xrightarrow{\mathrm{ICO}}\ket{\psi_1}\right) =0.83$. Now consider another state $\ket{\phi} = \sqrt{0.6} \ket{0}+\sqrt{0.4}\ket{1}$. We have $P\left(\ket{\psi_1}\otimes\ket{\phi}\xrightarrow{\mathrm{ICO}}\ket{\psi_2}\otimes\ket{\phi}\right) = 0.8$ and $P\left(\ket{\psi_2}\otimes\ket{\phi}\xrightarrow{\mathrm{ICO}}\ket{\psi_1}\otimes\ket{\phi}\right) =0.92$. Notice that $P\left(\ket{\psi_1}\xrightarrow{\mathrm{ICO}}\ket{\psi_2}\right) = 0.8$ is not increased by the use of $\ket{\phi}$. This is a consequence of our following proposition.
\begin{prop}
If, under the best strategy of ICO, $P\left(\ket{\psi_1}\xrightarrow{\mathrm{ICO}}\ket{\psi_2}\right)$ is equal to $|\psi_1^{d-1}|/|\psi_2^{d-1}|$, then this probability cannot be increased by the presence of any (catalyst) state. Here $\ket{\psi_1} = \sum_{i=0}^{d-1} \sqrt{\psi_1^i} \ket{i}$ and $\ket{\psi_2} = \sum_{i=0}^{d-1} \sqrt{\psi_2^i} \ket{i}$.
\end{prop}
\begin{mproof}
If, under the best strategy of ICO, $P(\ket{\psi_1}\rightarrow\ket{\psi_2})$ is equal to $|\psi^{d-1}_1|/|\psi^{d-1}_2|$,
then for any catalyst state $\ket{\phi}=\sum^{m}_{i=1}\sqrt{\phi_i}\ket{i}$,
the minimal coefficients of $\ket{\psi_1}\otimes\ket{\phi}$ and $\ket{\psi_2}\otimes\ket{\phi}$ are $\sqrt{\psi^{d-1}_1 \phi_{m}}$ and
$\sqrt{\psi^{d-1}_2 \phi_{m}}$, respectively. Thus, $P(\ket{\psi_1}\otimes\ket{\phi}\rightarrow\ket{\psi_2}\otimes\ket{\phi}
)=\min_{l\in[0,(d-1)m]}\frac{C_l(\psi_1\otimes\phi)}{C_l(\psi_2\otimes\phi)}\leq \frac{C_{d-1,m}(\psi_1\otimes\phi)}{C_{d-1,m}(\psi_2\otimes\phi)}=\frac{|\psi^{d-1}_1 \phi_{m}|}{|\psi^{d-1}_2 \phi_{m}|}=|\psi^{d-1}_1|/|\psi^{d-1}_2|$.
\end{mproof}
We note that the above proposition can be strengthened and we provide the necessary and sufficient conditions for the enhancement of the optimal probability for transformations under incoherent operations in the presence of catalysts as our next proposition.
\begin{prop}
For two pure states $\ket{\psi_1}=\sum^{d-1}_{i=0}\sqrt{\psi_1^i}\ket{i}$ and $\ket{\psi_2}=\sum^{d-1}_{i=0}\sqrt{\psi_2^i}\ket{i}$ there exists a catalyst $\ket{\phi}$ such that $P(\ket{\psi_1}\otimes\ket{\phi}\xrightarrow{\mathrm{ICO}}\ket{\psi_2}\otimes\ket{\phi})>P(\ket{\psi_1}\xrightarrow{\mathrm{ICO}}\ket{\psi_2})$, if and only if
\begin{eqnarray*}
P\left(\ket{\psi_1}\xrightarrow{\mathrm{ICO}}\ket{\psi_2}\right)< \min\left\{\frac{|\psi^{d-1}_1|}{|\psi^{d-1}_2|}, 1\right\}.
\end{eqnarray*}
\end{prop}
\begin{mproof}
The proof follows directly from Theorem \ref{th:DBG} of main text and Theorem \ref{th:Feng2005} of the appendix.
\end{mproof}

\section{Necessary and sufficient conditions for deterministic catalytic coherence transformations}
\label{sec:nec-suff}
We know that under incoherent operations, in the absence of catalysts, the necessary and sufficient conditions for transforming a pure state $\ket{\psi_1}$ to another pure state $\ket{\psi_2}$ are given by Theorem \ref{th:DBG}, i.e., $C_r(\ket{\psi_1}) \geq C_r(\ket{\psi_2})$. Here $C_r(\ket{\psi})$ denotes the relative entropy of coherence of $\ket{\psi}$ \cite{Baumgratz2014}. Now if we allow for catalysts, does the decrease of relative entropy of coherence, $C_r(\ket{\psi_1}) \geq C_r(\ket{\psi_2})$, ensure existence of an incoherent operation that maps $\ket{\psi_1}$ to $\ket{\psi_2}$? In the following we prove that this is not the case, i.e., mere decrease of the relative entropy of coherence is not sufficient. We next characterize the necessary and sufficient conditions for catalytic coherence transformations between the initial state $\ket{\psi_1}$ and the target state $\ket{\psi_2}$.
\begin{prop}\label{ne-sf-c}
For two pure states $\ket{\psi_1},\ket{\psi_2}\in \mathcal{H}(d)$, if the coefficients of $\ket{\psi_1}$, in a fixed basis, are all nonzero, then the necessary and sufficient conditions for catalytic coherence transformations are the simultaneous
decrease of a family of R\'enyi entropies which are defined as $S_\alpha(\psi^\md)=\mathrm{sgn}(\alpha)\ln\left( \mathrm{Tr}\left[\left(\psi^\md\right)^\alpha\right]\right)/(1-\alpha)$. Here $\psi^\md$ is the diagonal part of the pure state $\ket{\psi}$ and $\mathrm{sgn}(\alpha)=1$ for $\alpha\geq 0$, and $\mathrm{sgn}(\alpha)=-1$ when $\alpha< 0$. More precisely, there exists a catalyst state $\ket{\phi}$ such that $ \ket{\psi_1}\otimes\ket{\phi}\xrightarrow{\mathrm{ICO}}\ket{\psi_2}\otimes\ket{\phi}$ if and only if the conditions
\begin{align}
\frac{\tilde{S}_\alpha\left(\psi^\md_2\right)}{|\alpha|}&<\frac{\tilde{S}_\alpha\left(\psi^\md_1\right)}{|\alpha|}
\end{align}
are satisfied simultaneously for all $\alpha\in (-\infty, +\infty)$, where $\tilde{S}_\alpha\left(\psi^\md\right)=S_\alpha\left(\psi^\md\right)-\ln d$. For $\alpha=0$, $\tilde{S}_\alpha\left(\psi^\md\right)/|\alpha|=\lim_{\alpha\rightarrow 0^+}\tilde{S}_\alpha\left(\psi^\md\right)/|\alpha|=\sum^d_{i=1}\ln\psi^\md_i/d$ where $\psi^\md_i$ are components of $\psi^\md$ in a fixed reference basis.
\end{prop}
\begin{mproof}
We use a result from Ref. \cite{Turgut2007} (which we restate as Lemma \ref{lem:turgut} in appendix for clarity and completeness) to prove our proposition. For $\alpha \neq \{0,1\}$ note that $\tilde{S}_\alpha\left(\psi^\md\right)/|\alpha|=\ln A_\alpha\left(\psi^\md\right)/{(1-\alpha)} +\ln d/\alpha(1-\alpha)-\ln d/|\alpha|$, where $A_\alpha\left(\psi^\md\right) = \left(\frac{1}{d}\sum^{d}_{i=1}{\left(\psi_i^\md\right)}^\alpha\right)^{1/\alpha}$ (as in Lemma \ref{lem:turgut}). So for $\alpha \neq \{0,1\}$, the proof of our proposition follows from Lemma \ref{lem:turgut} and Theorem \ref{th:DBG}. Similarly, for $\alpha =1$, the proof follows from Lemma \ref{lem:turgut} and Theorem \ref{th:DBG}. For $\alpha=0$ the proof follows again by noting that

\begin{align*}
\lim_{\alpha\rightarrow0^+}\frac{\tilde{S}_\alpha\left(\psi_1^{\md}\right)}{|\alpha|}= \frac{1}{d}\sum_{i=1}^d \ln \psi_1^i = \ln A_0(\psi_1^{\md}),
\end{align*}
where, for any probability vector $p$, $A_0(p)=\left(\prod_{i=1}^{d}p_i\right)^{\frac{1}{d}}$, as in Lemma \ref{lem:turgut} and $\psi_1^{\md}=\left(\psi_1^1,\ldots,\psi_1^d\right)^T$. This completes the proof of proposition.
\end{mproof}

We emphasize that Proposition \ref{ne-sf-c} assumes that the initial state $\ket{\psi_1}$ must contain only nonzero entries. But this problem can be remedied by allowing slight perturbation to the initial state. Moreover, the strict inequality in Proposition \ref{ne-sf-c} can be made nonstrict. In this view, we generalize Proposition \ref{ne-sf-c} to the following proposition.

\begin{prop}\label{gn-sf-nc}
For two pure states $\ket{\psi_1}$ and $\ket{\psi_2}$, the following two conditions are equivalent:\\
1. For a given pure state $\ket{\psi_1}$ there exists a state $\ket{\psi^\varepsilon_1}$ with $\varepsilon>0$ and a catalyst state $\ket{\phi}$ such that (i) $||\ket{\psi_1}-\ket{\psi^\varepsilon_1}||<\varepsilon$; (ii) $\ket{\psi^\varepsilon_1}\otimes\ket{\phi}\xrightarrow{\mathrm{ICO}} \ket{\psi_2}\otimes\ket{\phi}$.\\
2. For all  $\alpha\in (-\infty, +\infty)$
\begin{eqnarray}
\frac{\tilde{S}_\alpha\left(\psi^\md_2\right)}{|\alpha|}\leq\frac{\tilde{S}_\alpha\left(\psi^\md_1\right)}{|\alpha|}.
\end{eqnarray}
\end{prop}

\begin{mproof}
We first prove the implication $1\Rightarrow2$. Although $\ket{\psi^\varepsilon_1}$ may have zero component, there always exists a state $\ket{\psi^{\varepsilon'}_1}$ close to $\ket{\psi^\varepsilon_1}$ with nonzero components only, which also satisfy (i) and (ii) in the condition 1. Thus, without loss of any generality, we can assume the components of $\ket{\psi^\varepsilon_1}$ are all nonzero. Since $\ket{\psi^\varepsilon_1}\otimes\ket{\phi}\xrightarrow{\mathrm{ICO}} \ket{\psi_2}\otimes\ket{\phi}$, then by Proposition \ref{ne-sf-c},
\begin{eqnarray*}
\frac{\tilde{S}_\alpha\left(\psi^\md_2\right)}{|\alpha|}<\frac{\tilde{S}_\alpha\left((\psi^\varepsilon_1)^\md\right)}{|\alpha|}
\end{eqnarray*}
for every $\varepsilon >0$. Based on the continuity of functions $\tilde{S}_\alpha(\cdot)/|\alpha|$, we have $\tilde{S}_\alpha\left(\psi^\md_2\right)/|\alpha|\leq \tilde{S}_\alpha\left(\psi^\md_1\right)/|\alpha|$.

Now we prove the implication $2\Rightarrow1$. For $\ket{\psi_1}=\sum^d_{i=1}\sqrt{\psi^i_1}\ket{i}$ let $\ket{\psi^\varepsilon_1}=\sum^d_{i=1}\sqrt{(1-\varepsilon)\psi^i_1+\varepsilon/d}\ket{i}$. Then
$||\ket{\psi_1}-\ket{\psi^\varepsilon_1}||\rightarrow 0$, when $\varepsilon\rightarrow 0$. Due to Lemma \ref{Shurconv} of appendix, we know
the functions $\tilde{S}_\alpha(\cdot)/|\alpha|$ are strictly Schur concave for all $\alpha\in(-\infty, \infty)$, thus $\tilde{S}_\alpha\left(\psi^\md_1\right)/|\alpha|<\tilde{S}_\alpha\left((\psi^\varepsilon_1)^\md\right)/|\alpha|$ for every $\varepsilon >0$. As $\tilde{S}_\alpha\left(\psi^\md_2\right)/|\alpha|\leq \tilde{S}_\alpha\left(\psi^\md_1\right)/|\alpha|$, we have
\begin{eqnarray*}
\frac{\tilde{S}_\alpha\left(\psi^\md_2\right)}{|\alpha|}<\frac{\tilde{S}_\alpha\left((\psi^\varepsilon_1)^\md\right)}{|\alpha|}.
\end{eqnarray*}
It is easy to see the coefficients of $\ket{\psi^\varepsilon_1}$ are all nonzero. Hence, by Proposition \ref{ne-sf-c}, there exists a  catalyst in state $\ket{\phi}$ such that $\ket{\psi^\varepsilon_1}\otimes\ket{\phi}\xrightarrow{\mathrm{ICO}} \ket{\psi_2}\otimes\ket{\phi}$. This completes the proof.
\end{mproof}
To elaborate more about the necessary and sufficient conditions for catalytic coherence transformations we consider various examples. Fig. \ref{figure2} shows that for states  $\ket{\psi_{1}} = \sqrt{0.4}\ket{0}+ \sqrt{0.4}\ket{1}+ \sqrt{0.1}\ket{2}+ \sqrt{0.1}\ket{3}$ and $\ket{\psi_{2}} = \sqrt{0.5}\ket{0}+ \sqrt{0.25}\ket{1}+ \sqrt{0.25}\ket{2}$ a catalytic transformation is possible but for states $\ket{\psi_{1}} = \sqrt{0.5}\ket{0}+ \sqrt{0.4}\ket{1}+ \sqrt{0.1}\ket{2}$ and $\ket{\psi_{2}} = \sqrt{0.6}\ket{0}+\sqrt{0.25}\ket{1}+ \sqrt{0.15}\ket{2}$ no catalytic transformation is possible.
Moreover, using the similar techniques as in Ref. \cite{Brandao2015}, we can remove the $(-\infty, 0)$ part with the help of another ancillary qubit.
\begin{prop}
\label{prop:rem-zero}
For pure states $\ket{\psi_1}$ and $\ket{\psi_2}$, the following two conditions are equivalent:\\
1. For a given pure state $\ket{\psi_1}$ there exist states $\ket{\psi^\varepsilon_1}\otimes\ket{0^\varepsilon}$ with $\varepsilon>0$ and $\ket{\phi}$ such that
(i) $||\ket{\psi_1}\otimes\ket{0}-\ket{\psi^\varepsilon_1}\otimes\ket{0^\varepsilon}||<\varepsilon$;
(ii) $\ket{\psi^\varepsilon_1}\otimes\ket{0^\varepsilon}\otimes\ket{\phi}\xrightarrow{\mathrm{ICO}} \ket{\psi_2}\otimes\ket{0}\otimes\ket{\phi}$.\\
2. For all  $\alpha\in [0, +\infty)$
\begin{eqnarray}
\frac{\tilde{S}_\alpha\left(\psi^\md_2\right)}{|\alpha|}\leq\frac{\tilde{S}_\alpha\left(\psi^\md_1\right)}{|\alpha|}.
\end{eqnarray}
\end{prop}
\begin{mproof}
($1 \Rightarrow 2$) Similar to the proof of Proposition \ref{gn-sf-nc},  without loss of generality, we can assume the components of $\ket{\psi^\varepsilon_1}$ and $\ket{0^\varepsilon}$  are all nonzero. Since $\ket{\psi^\varepsilon_1}\otimes\ket{0^\varepsilon}\otimes\ket{\phi}\xrightarrow{\mathrm{ICO}} \ket{\psi_2}\otimes\ket{0}\otimes\ket{\phi}$, then by Proposition \ref{ne-sf-c},
\begin{eqnarray*}
\frac{\tilde{S}_\alpha\left((\proj{\psi_2}\otimes \proj{0})^\md\right)}{|\alpha|}<\frac{\tilde{S}_\alpha\left((\proj{\psi^\varepsilon_1}\otimes \proj{0^\varepsilon})^\md\right)}{|\alpha|}
\end{eqnarray*}
 for every $\varepsilon >0$. Based on the continuity and additivity of $\frac{\tilde{S}_\alpha(\cdot)}{|\alpha|}$ we have $\frac{\tilde{S}_\alpha\left(\psi^\md_2\right)}{|\alpha|}\leq\frac{\tilde{S}_\alpha\left(\psi^\md_1\right)}{|\alpha|}$.\\
 ($2 \Rightarrow 1$)
For $\ket{\psi_1}=\sum^d_{i=1}\sqrt{\psi^i_1}\ket{i}$ let $\ket{\psi^\varepsilon_1}=\sum^d_{i=1}\sqrt{(1-\varepsilon)\psi^i_1+\varepsilon/d}\ket{i}$,
and $\ket{0^\varepsilon}=\sqrt{1-\varepsilon/2}\ket{0}+\sqrt{\varepsilon/2}\ket{1}$.
 Then
$||\ket{\psi_1}\otimes\ket{0}-\ket{\psi^\varepsilon_1}\otimes\ket{0^\varepsilon}||\rightarrow 0$, when $\varepsilon\rightarrow 0$. Similar to proof of Proposition \ref{gn-sf-nc}, it is easy to obtain
\begin{align}\label{ineq_Sa}
\frac{\tilde{S}_\alpha\left((\proj{\psi_2}\otimes \proj{0})^\md\right)}{|\alpha|}<\frac{\tilde{S}_\alpha\left((\proj{\psi^\varepsilon_1}\otimes \proj{0^\varepsilon})^\md\right)}{|\alpha|}
\end{align}
for all  $\alpha\in [0, +\infty)$. In the case of  $\alpha<0$, due to the definition of $\tilde{S}_\alpha$, the
left side of inequality \ref{ineq_Sa} will be $-\infty$, and the right side will be finite.
By Proposition \ref{ne-sf-c}, there exist catalyst state $\ket{\phi}$ such that $\ket{\psi^\varepsilon_1}\otimes\ket{0^\varepsilon}\otimes\ket{\phi}\xrightarrow{\mathrm{ICO}} \ket{\psi_2}\otimes\ket{0}\otimes\ket{\phi}$. This completes the proof.
\end{mproof}

In fact, we need not worry about the rank of the initial state. If $\ket{\psi_1}$ can be
transformed to $\ket{\psi_2}$ using the catalyst $\ket{\phi}$, then
$C_s(\ket{\psi_2}\otimes\ket{\phi})\leq C_s(\ket{\psi_1}\otimes\ket{\phi})$ where $C_s(\ket{\psi}):=\mathrm{Rank}(\psi^\md)$ and is a proper measure of coherence \cite{Bu2015}. $\psi^\md$ is the diagonal part of $\ket{\psi}$ in the fixed reference basis. This implies $C_s(\ket{\psi_2})\leq C_s(\ket{\psi_1})$, i.e., $\mathrm{Rank}(\psi^\md_2)\leq\mathrm{Rank}(\psi^\md_1)$. Therefore, $\ket{\psi_1}$ and $\ket{\psi_2}$ can also be viewed as pure states in $\mathcal{H}(d')$, where $d'=\max \{C_s(\psi_1), C_s(\psi_2)\}=C_s(\psi_1)$; $\ket{\psi_1}$ will be full rank and the above propositions can be used.

\begin{figure}
\subfigure[]{
\includegraphics[width=40 mm]{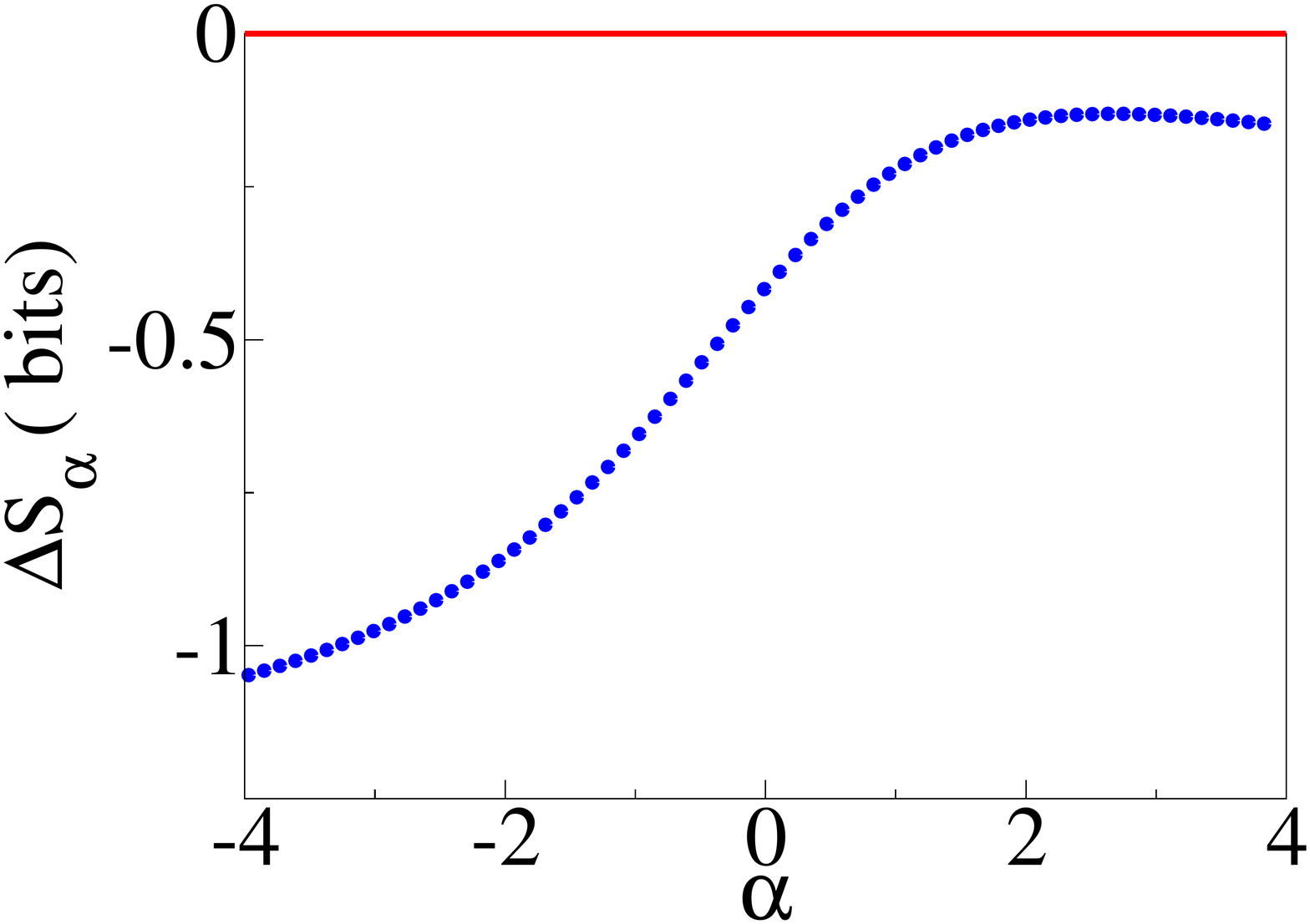}{\label{fig4}}}
\subfigure[]{
\includegraphics[width=40 mm]{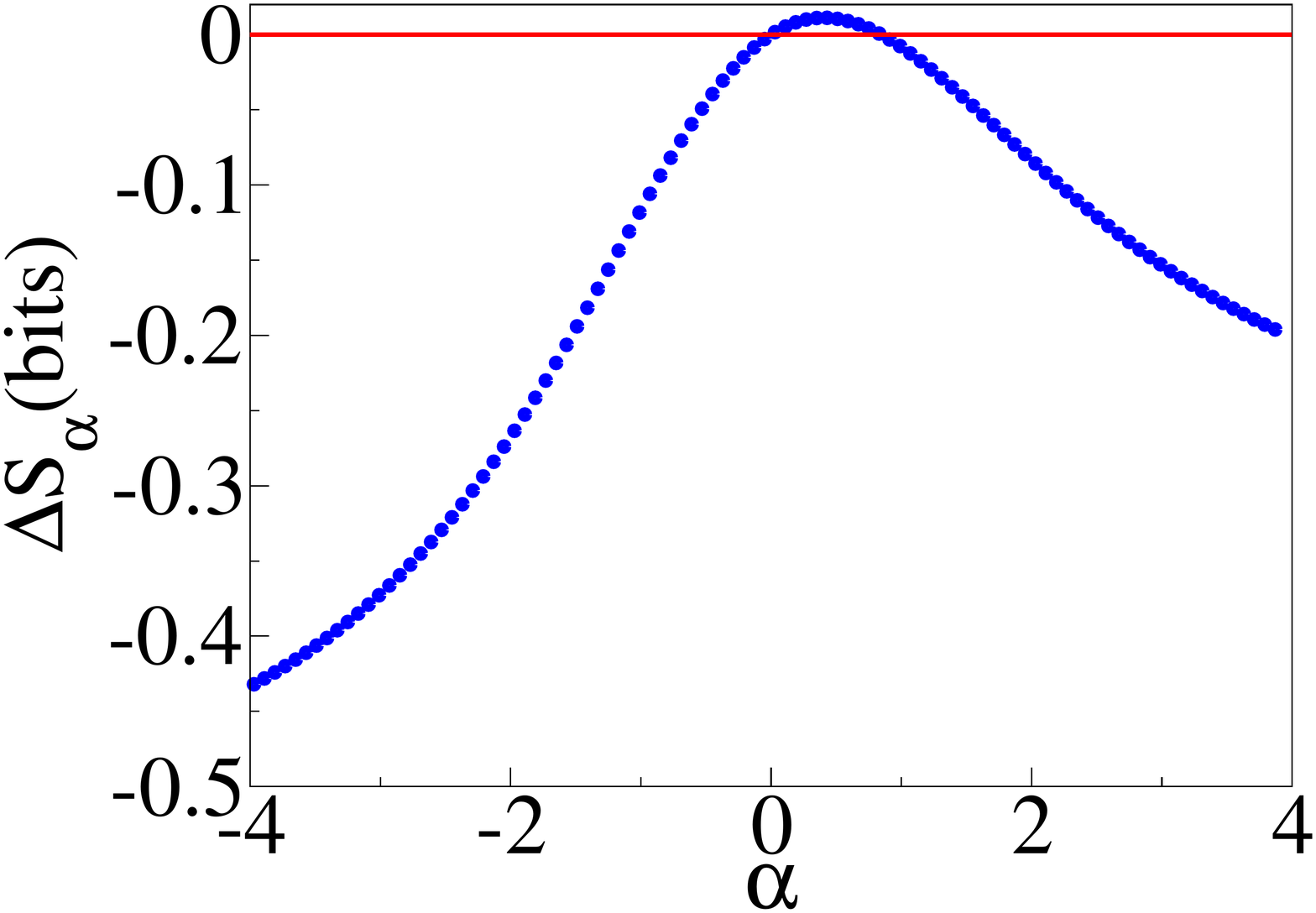}{\label{fig2}}}
\caption{(Color online) The plot shows the variation of $\Delta \tilde{S}_\alpha = \tilde{S}_\alpha(\psi^\md_{2}) - \tilde{S}_\alpha(\psi^\md_{1})$ as a function of $\alpha$. In Fig. \ref{fig4}, we take $\ket{\psi_{1}} = \sqrt{0.4}\ket{0}+ \sqrt{0.4}\ket{1}+ \sqrt{0.1}\ket{2}+ \sqrt{0.1}\ket{3}$ and $\ket{\psi_{2}} = \sqrt{0.5}\ket{0}+ \sqrt{0.25}\ket{1}+ \sqrt{0.25}\ket{2}$. Based on our Proposition \ref{ne-sf-c}, the transformation from $\ket{\psi_1}$ to $\ket{\psi_2}$  will be possible with the aid of a catalyst. In Fig. \ref{fig2}, we take $\ket{\psi_{1}} = \sqrt{0.5}\ket{0}+ \sqrt{0.4}\ket{1}+ \sqrt{0.1}\ket{2}$ and $\ket{\psi_{2}} = \sqrt{0.6}\ket{0}+\sqrt{0.25}\ket{1}+ \sqrt{0.15}\ket{2}$. Because for certain values of $\alpha$, $\Delta \tilde{S}_\alpha$ increases, from Proposition \ref{ne-sf-c}, there does not exist a catalyst that can allow the transformation from $\ket{\psi_1}$ to $\ket{\psi_2}$ in this case.}\label{figure2}
\end{figure}

All the Propositions \ref{ne-sf-c}, \ref{gn-sf-nc}, and \ref{prop:rem-zero} tell us that in order to check whether the transformation under consideration is possible, we need to check infinitely many conditions, thus, making the proposition only of theoretical merit. However, as we show below in Proposition \ref{prop:two-con}, if only a few conditions on R\'enyi entropy hold, then they suffice to show that all other conditions hold automatically.
\begin{prop}
\label{prop:two-con}
Consider two pure states $\ket{\psi_1}$ and $\ket{\psi_2}$. Given $\varepsilon>0$, we can construct two pure states $\ket{\phi_1}$ and $\ket{\phi_2}$ with $\phi^\md_1\in B_\varepsilon(\psi^\md_1)$, $\phi^\md_2\in B_\varepsilon(\psi^\md_2)$. Here, for two probability vectors $x$ and $y$, $B_\varepsilon(x)$ is the $\varepsilon$ ball around $x$ and is defined as $B_\varepsilon(x):=\set{y:\frac{1}{2}\sum_i|y_i-x_i|<\varepsilon}$. Consider following conditions:
\begin{subequations}
\begin{align}
   &S_0(\psi^\md_2)\leq S_0(\phi^\md_1)+\frac{\log\varepsilon}{1-\alpha}  ~~~(\text{for~} 0<\alpha<1); \label{eq:condi}\\
   &S_\infty(\phi^\md_2)-\frac{\log\varepsilon}{\alpha-1}\leq S_\infty(\psi^\md_1) ~~~ (\text{for~} \alpha>1);\label{eq:condi1}\\
   &\frac{\tilde{S}_\alpha(\psi^\md_2)}{|\alpha|}\leq \frac{\tilde{S}_\alpha(\psi^\md_1)}{|\alpha|}  ~~~(\text{for~} \alpha=0).\label{eq:condi11}
  \end{align}
 \end{subequations}
If conditions (\ref{eq:condi}), (\ref{eq:condi1}), and (\ref{eq:condi11}) hold, then for any $\alpha\in[0,\infty)$,
\begin{eqnarray}
\frac{\tilde{S}_\alpha(\psi^\md_2)}{|\alpha|}\leq\frac{\tilde{S}_\alpha(\psi^\md_1)}{|\alpha|}.
\end{eqnarray}
\end{prop}
\begin{proof}
Based on Lemma \ref{lem:ineq_con} of appendix, we can construct two pure states $\ket{\phi_1}$ and
$\ket{\phi_2}$ with $\phi^\md_1\in B_\varepsilon(\psi^\md_1)$, $\phi^\md_2\in B_\varepsilon(\psi^\md_2)$, such that
\begin{subequations}
\begin{align}
   &S_\alpha(\psi^\md_1)\geq S_0(\phi^\md_1)+\frac{\log\varepsilon}{1-\alpha} ~~~ (\text{for~} 0<\alpha<1);\label{eq:condi2}\\
   &S_\infty(\phi^\md_2)-\frac{\log\varepsilon}{\alpha-1}\geq S_\alpha(\psi^\md_2)  ~~~ (\text{for~} \alpha>1).\label{eq:condi22}
  \end{align}
  \end{subequations}
Note that for any probability vector $x$, $S_\alpha(x)\leq S_\beta(x)$ if $\alpha\geq\beta$. Now for $0<\alpha<1$, from conditions (\ref{eq:condi2}) and (\ref{eq:condi}), we have
\begin{align}
\tilde{S}_\alpha(\psi^\md_1)&=S_\alpha(\psi^\md_1)-\ln d\nonumber\\
&\geq S_0(\phi^\md_1)+\frac{\log\varepsilon}{1-\alpha}-\ln d\nonumber\\
&\geq S_0(\psi^\md_2)-\ln d\nonumber\\
&\geq S_\alpha(\psi^\md_2)-\ln d=\tilde{S}_\alpha(\psi^\md_2).
\end{align}
Similarly, for $\alpha>1$, from conditions (\ref{eq:condi1}) and (\ref{eq:condi22}), we have
\begin{align}
 \tilde{S}_\alpha(\psi^\md_1) &= S_\alpha(\psi^\md_1)-\ln d\nonumber\\
 &\geq S_\infty(\psi^\md_1) - \ln d\nonumber\\
 &\geq S_\infty(\phi^\md_2)-\frac{\log\varepsilon}{\alpha-1}- \ln d\nonumber\\
 &\geq \tilde{S}_\alpha(\psi^\md_2).
 \end{align}
Combining the above two equations with the condition (\ref{eq:condi11}), we get
\begin{eqnarray*}
\frac{\tilde{S}_\alpha(\psi^\md_2)}{|\alpha|}\leq\frac{\tilde{S}_\alpha(\psi^\md_1)}{|\alpha|},
\end{eqnarray*}
for all values of $\alpha\in[0,\infty)$, where $\alpha=1$ case comes from the continuity of $\tilde{S}_\alpha(\cdot)/|\alpha|$.
\end{proof}
Hence, we only need to check conditions (\ref{eq:condi}), (\ref{eq:condi1}), and (\ref{eq:condi11}) to determine whether the transformation between two pure states is possible with the aid of catalysts. This establishes the practicality of Propositions \ref{ne-sf-c}, \ref{gn-sf-nc}, and \ref{prop:rem-zero}.

\section{Entanglement assisted coherence transformations}
\label{sec:b-cat-coh}
Consider a pair of pure states such that there exists no catalytic incoherent transformation (see Fig. \ref{figure2}) between them. Can we find an incoherent operation between such pair of states with some assistance of another physical resource? The following proposition answers this question. We follow the proof techniques of Ref. \cite{Mueller2015} to prove the proposition.
\begin{prop}\label{cohcr}
For any pure states $\ket{\psi_1}$ and $\ket{\psi_2}$ there exist a $k$-partite pure state $\ket{\tilde{\phi}}_{1,\ldots,k}$ and $\ket{\phi_1}\otimes\ldots\otimes\ket{\phi_k}$ such that
\begin{eqnarray*}
\ket{\psi_1}\otimes\ket{\tilde{\phi}}_{1,\ldots,k}\xrightarrow{\mathrm{ICO}} \ket{\psi_2}\otimes\ket{\phi_1}\otimes\ldots\otimes\ket{\phi_k}
\end{eqnarray*}
with $\phi^\md_i=\tilde{\phi}_i^\md:=\mathrm{Tr}_{\set{1,\ldots,k}/i}\tilde{\phi}^\md$ and $k\leq 3$ if and only if the following two conditions are satisfied: (1)
$C_s(\ket{\psi_2})\leq C_s(\ket{\psi_1})$  and (2) $C_r(\ket{\psi_2})<C_r(\ket{\psi_1})$. Here $C_s$ is a proper coherence measure defined in \cite{Bu2015}, which for a pure state is equal to the number of nonzero coefficients in the state spanned in the reference basis.
\end{prop}
\begin{mproof} Note that from Theorem \ref{th:DBG}, $\ket{\psi_1}\otimes\ket{\tilde{\phi}}_{1,\ldots,k}\xrightarrow{\mathrm{ICO}} \ket{\psi_2}\otimes\ket{\phi_1}\otimes\ldots\otimes\ket{\phi_k}$ is equivalent to $\psi^{(d)}_1\otimes\tilde{\phi}^\md_{1,\dots,k}\prec \psi^{(d)}_2\otimes\phi^\md_1\otimes\ldots\otimes\phi^\md_k$. Then the proof of our proposition follows from Lemma \ref{lem:Muller} of the appendix.
\end{mproof}
Now let us apply Proposition \ref{cohcr} to a numerical example. Consider two pure states $\ket{\psi_1}$ and $\ket{\psi_2}$ with $\psi^\md_{1}=(0.5,0.4,0.1)$ and $\psi^\md_{2}=(0.6,0.25,0.15)$. Then, we know that $\ket{\psi_1}$ cannot be transformed to $\psi_2$ using any catalyst as there exists an $\alpha$ such that $\tilde{S}_\alpha\left(\psi^\md_{2}\right)-\tilde{S}_\alpha\left(\psi^\md_{1}\right)<0$ (see Fig. \ref{figure2}). Thus, $\ket{\psi_1}\otimes\ket{\phi}^{\otimes n}\rightarrow\ket{\psi_2}\otimes\ket{\phi}^{\otimes n}$ is not possible. However, $\ket{\psi_1}$ can  be transformed to $\ket{\psi_2}$ using an entanglement assisted incoherent transformation as $C_r(\ket{\psi_2})<C_r(\ket{\psi_1})$ and $C_s(\ket{\psi_2})=C_s(\ket{\psi_1})$ (see Proposition \ref{cohcr}). We emphasize here that in the above process of entanglement assisted incoherent transformation, coherence in the ancillary system is not consumed. This can be proved by using the following fact
\begin{align*}
&C_r(\ket{\phi_1}\otimes\dots\otimes\ket{\phi_k})-C_r(\ket{\tilde{\phi}}_{1,\dots,k})\\
&~~~~~~~~=\sum^k_{i=1}S\left(\phi^\md_i\right)-S\left(\tilde{\phi}^\md_{1,\dots,k}\right)\geq0.
\end{align*}
Particularly, when $k=2$, for pure states $\ket{\tilde{\phi}}_{12}$, $\ket{\phi}_1$ and $\ket{\phi}_2$, where $\phi^\md_1=\tilde{\phi}^\md_1, \phi^\md_2=\tilde{\phi}^\md_2$, we have $ C_r\left(\ket{\phi}_1\otimes \ket{\phi}_2\right)-C_r\left(\ket{\tilde{\phi}}_{12}\right)=I\left(\left(\ket{\tilde{\phi}}_{12}\right)^\md\right)\leq I\left(\ket{\tilde{\phi}}_{12}\right)= 2E_r\left(\ket{\tilde{\phi}}_{12}\right)$, giving an upper bound on the increased coherence. Here, $I(\rho_{12}) := S(\rho_1)+ S(\rho_2)-S(\rho_{12})$ is the mutual information of $\rho_{12}$ and $E_r\left(\ket{\tilde{\phi}}_{12}\right):=S\left(\mathrm{Tr}_2\left[\ket{\tilde{\phi}}\bra{\tilde{\phi}}_{12}\right]\right)$ is the entropy of entanglement of the state $\ket{\tilde{\phi}}_{12}$. Further, we generalize Proposition \ref{cohcr} to the following proposition. The proof techniques for the following proposition are adapted from Refs. \cite{Brandao2015, LostaglioB2015}.
\begin{prop}
For two pure states $\ket{\psi_1}$ and $\ket{\psi_2}$, the following two conditions are equivalent:\\
1. For a given state $\ket{\psi_1}$ there exist states $\ket{\psi^\varepsilon_1}$, $\ket{\tilde{\phi}}_{1,\ldots,k}$ and $\ket{\phi_1}\otimes\ldots\otimes\ket{\phi_k}$ with $\tilde{\phi}^\md_i=\phi_i^\md$ and $\varepsilon>0$ such that (i) $||\ket{\psi_1}-\ket{\psi^\varepsilon_1}||<\varepsilon$; (ii) $\ket{\psi^\varepsilon_1}\otimes\ket{\tilde{\phi}}_{1,\ldots,k}\xrightarrow{\mathrm{ICO}} \ket{\psi_2}\otimes\ket{\phi_1}\otimes\ldots\otimes\ket{\phi_k}$. Here, $k\leq 3$.\\
2. $C_r(\ket{\psi_2})\leq C_r(\ket{\psi_1})$.
\end{prop}

\begin{mproof}
Let us first prove the implication $1\Rightarrow2$. Since $\ket{\psi^\varepsilon_1}\otimes\ket{\tilde{\phi}}_{1,\ldots,k}\xrightarrow{\mathrm{ICO}} \ket{\psi_2}\otimes\ket{\phi_1}\otimes\ldots\otimes\ket{\phi_k}$ and coherence cannot increase under incoherent operations, we have $C_r(\ket{\psi_2}\otimes\ket{\phi_1}\otimes\ldots\otimes\ket{\phi_k})\leq C_r(\ket{\psi^\varepsilon_1}\otimes\ket{\tilde{\phi}}_{1,\ldots,k})$. As $C_r(\ket{\phi_1}\otimes\ldots\otimes\ket{\phi_k})-C_r(\ket{\tilde{\phi}}_{1,\ldots,k})=\sum^k_{i=1}S\left(\phi^\md_i\right)-S\left(\tilde{\phi}^\md_{1,\ldots,k}\right)\geq0$, we have $C_r(\ket{\psi_2})\leq C_r(\ket{\psi^\varepsilon_1})$. Let $\varepsilon\rightarrow 0$, then $C_r(\ket{\psi_2})\leq C_r(\ket{\psi_1})$.

The implication $2\Rightarrow1$ can be proved as follows. For $\ket{\psi_1}=\sum^d_{i=1}\sqrt{\psi^i_1}\ket{i}$ let $\ket{\psi^\varepsilon_1}=\sum^d_{i=1}\sqrt{(1-\varepsilon)\psi^i_1+\varepsilon/d}\ket{i}$. Then
$||\ket{\psi_1}-\ket{\psi^\varepsilon_1}||\rightarrow 0$, when $\varepsilon\rightarrow 0$. Moreover,
$C_r(\ket{\psi_2})\leq C_r(\ket{\psi_1})$ is equivalent to $S\left(\psi^\md_2\right)\leq S\left(\psi^\md_1\right)$.
As $S(\cdot)$ is strictly concave, then $S\left(\psi^\md_2\right)\leq S\left(\psi^\md_1\right) <S\left(\tilde{\psi}^\md_1\right)$. It is easy to see
$C_s(\ket{\psi^\varepsilon_1})=d\geq C_s(\ket{\psi_2})$. Now using Proposition \ref{cohcr} we complete the proof.
\end{mproof}

\section{Conclusion and Outlook}
\label{sec:conclusion}
In this work we find the necessary and sufficient conditions for the deterministic and stochastic coherence transformations between pure quantum states mediated by catalysts using only incoherent operations. We first find the necessary and sufficient conditions for the possibility of the increase of the optimal probability of achieving an otherwise impossible transformation with the aid of a catalyst. Then, we show that for a given pair of pure quantum states, the necessary and sufficient conditions for a deterministic catalytic transformation from one state to another are the simultaneous decrease of a family of R\'enyi entropies of the corresponding diagonal parts of the given pure states in a fixed reference basis. We also discuss about the practicality of these conditions. Further, we delineate the structure of the catalysts and find that it is possible for a pure quantum state to act as a catalyst for itself for a given otherwise impossible state transformation using incoherent operations. This phenomena may be termed as {\it self catalysis}. Moreover, for the pair of states which violate the necessary and sufficient conditions for deterministic coherence transformations, we consider the possibility of using an entangled state and show that even though there exists no catalyst for such a pair of states but an entangled state can indeed be used to facilitate the transformation. We dub such transformations as entanglement assisted coherence transformations. Here we emphasize that in entanglement assisted coherence transformations the coherence of the entangled state is {\it not} consumed at all. We also provide necessary and sufficient conditions for the entanglement assisted coherence transformations. In this way we completely characterize the allowed manipulations of the coherence of pure quantum states and thus, our work contributes towards a complete resource theory of coherence based on incoherent operations.

The consideration of catalytic transformations is very natural and has resulted in strikingly nontrivial consequences. One such instance is the introduction of many second laws of quantum thermodynamics superseding the common wisdom of single second law of macroscopic thermodynamics. Now given the importance of quantum coherence in quantum thermodynamics and various other avenues, we hope that our results which provide the limitations on coherence transformations will be extremely helpful in the processing of quantum coherence in such situations and in particular, in the context of {\it single-shot} quantum  information theory. Further, it will be important to analyze the possibility of {\it self catalysis} in greater detail in future as the catalysts in this case are readily available.

\smallskip
\noindent
\begin{acknowledgments}
 U.S. acknowledges the research fellowship of Department of Atomic Energy, Government of India and thanks Avijit Misra, Manabendra Nath Bera and Samyadeb Bhattacharya for useful discussions on the catalytic coherence. This  project is supported by National Natural Science Foundation of China (11171301, 11571307) and by the Doctoral Programs Foundation of the Ministry of Education of China (J20130061).
\end{acknowledgments}

\appendix
\section{Some useful earlier results}
\label{appendix:1}
Here, for the sake of completeness, we restate various results obtained earlier by other researchers which are useful for our work. The following theorem is due to  Ref. \cite{Sun2005}.

\begin{thm}[\cite{Sun2005}]\label{th:Sun2005}
For a bipartite qubit system with states $\ket{\psi_1}=\sum^4_{i=1}\sqrt{\alpha_i}\ket{ii}$, $\ket{\psi_2}=\sum^4_{i=1}\sqrt{\beta_i}\ket{ii}$ such that $\ket{\psi_1}\nrightarrow \ket{\psi_2}$ under local operations and classical communication (LOCC). Without loss of generality we can assume that the coefficients $\{\alpha_i\}$, $\{\beta_i\}$ are real and arranged in decreasing order. Then the necessary and sufficient conditions for the existence of a catalyst $\ket{\phi}=\sqrt{a}\ket{11}+\sqrt{1-a}\ket{22}$ $(a\in(0.5,1))$ for these two states are the following two conditions: $\alpha_1\leq \beta_1;~\alpha_1+\alpha_2>\beta_1+\beta_2;~\alpha_1+\alpha_2+\alpha_3\leq\beta_1+\beta_2+\beta_3$, and
\begin{align}
&\max\left\{\frac{\alpha_1+\alpha_2-\beta_1}{\beta_2+\beta_3},1-\frac{\alpha_4-\beta_4}{\beta_3-\alpha_3}\right\}\nonumber\\
&\leq a\leq \min\left\{\frac{\beta_1}{\alpha_1+\alpha_2},\frac{\beta_1-\alpha_1}{\alpha_2-\beta_2},1-\frac{\beta_4}{\alpha_3+\alpha_4}\right\}.
\end{align}
\end{thm}

The following theorem is from Refs. \cite{Feng2004, Feng2005}.
\begin{thm}[\cite{Feng2005}]\label{th:Feng2005}
For two $d$-dimensional probability vectors $p$ and $q$ with the components arranged in decreasing order, there exists a probability vector $r$ such that $P(p\otimes r \rightarrow q\otimes r)> P(p\rightarrow q)$ {\it if and only if}
\begin{eqnarray*}
P\left(p\rightarrow q\right)< \min\left\{\frac{p_d}{q_d}, 1\right\}.
\end{eqnarray*}
\end{thm}

In the case of catalytic majorization the necessary and sufficient conditions for catalytic transformations are obtained independently in Refs. \cite{Turgut2007} and \cite{Klimesh2007}. In Ref. \cite{Turgut2007}, the following result was obtained.
\begin{lem}[\cite{Turgut2007}]
\label{lem:turgut}
Let $p$ and $q$ be two distinct $d$-element probability vectors arranged in decreasing order with $p$ having nonzero elements. Then the existence of a vector $r$ such that $p\otimes r\prec q\otimes r$ is equivalent to the following three strict inequalities:
\begin{eqnarray}
A_\alpha(p)&>&A_\alpha(q) \text{~~for~} \alpha\in(-\infty, 1);\\
A_\alpha(p)&<&A_\alpha(q) \text{~~for~} \alpha\in(1, \infty);\\
S(p)&>&S(q).
\end{eqnarray}
where $A_\alpha:=(\frac{1}{d}\sum^{d}_{i=1}p^\alpha_i)^{\frac{1}{\alpha}}$, and $S(p)=-\sum^d_{i=1}p_i \log p_i$ is the Shannon entropy. For $\alpha=0$, $A_0(p)=(\prod p_i)^{1/d}$. If any component of vector $p$ is zero, then $A_\alpha=0$ for all
$\alpha\leq 0$.
\end{lem}

The following lemma is from  Ref. \cite{Brandao2015}.

\begin{lem}[\cite{Brandao2015}]
\label{Shurconv}
The R\'enyi entropies $S_\alpha$ are strictly Schur concave for $\alpha\in(-\infty, 0)\cup(0, \infty)$. The R\'enyi entropies for $\alpha=0, \pm\infty$ are Schur concave. Also, the function $\sum_i \log p_i$ is strictly Schur concave.
\end{lem}
Since $\tilde{S}_\alpha(\psi^\md)$ is equal to $S_\alpha(\psi^\md)-\ln d$ and $\lim_{\alpha\to 0^+}\tilde{S}_\alpha(\psi^\md)/|\alpha|=\frac{1}{d}\sum_i\ln\psi^\md_i$, the above lemma holds for functions $\tilde{S}_\alpha(\cdot)/|\alpha|$ too. That is, $\tilde{S}_\alpha(\cdot)/|\alpha|$ are strictly Schur concave for all $\alpha\in(-\infty, +\infty)$.

The following lemma is from  Ref. \cite{Mueller2015}.
\begin{lem}[\cite{Mueller2015}]
\label{lem:Muller}
Let $p$ and $q$ be $d$-dimensional probability vectors with components being arranged in decreasing order and $p\neq q$. Then there exists a $k$-partite probability distribution $r_{1,\ldots,k}$ such that
\begin{eqnarray*}
q \otimes r_{1,\ldots,k} \prec p\otimes(\otimes r_1\otimes \ldots\otimes r_k)
\end{eqnarray*}
if and only if $\mathrm{Rank}(p)\leq \mathrm{Rank}(q)$ and $S(p)<S(q)$. Here, we can always choose $k=3$. $S(p)=-\sum^d_{i=1}p_i \log p_i$ is the Shannon entropy.
\end{lem}
Consider a probability vector $x$. Define another subnormalized probability vector $x'$ from the $\varepsilon$ ball $B_\varepsilon(x)$ around $x$, defined as
\begin{eqnarray}
B_\varepsilon(x):=\set{y:\frac{1}{2}\sum_i|y_i-x_i|<\varepsilon}
\end{eqnarray}
for any $\varepsilon>0$. Now we have the following lemma from Ref. \cite{Brandao2015}.
\begin{lem}[\cite{Brandao2015}]\label{lem:ineq_con}
Given any probability vector $x$, for $0<\alpha<1$ and $\varepsilon>0 $, we can construct a probability vector $x'\in B_\varepsilon(x)$ such that
\begin{eqnarray}
S_\alpha(x)\geq S_0(x')+\frac{\log\varepsilon}{1-\alpha}.
\end{eqnarray}
For $\alpha>1$, we can construct another probability vector $x''\in B_\varepsilon(x)$ such that
\begin{eqnarray}
S_\infty(x'')-\frac{\log\varepsilon}{\alpha-1}\geq S_\alpha(x).
\end{eqnarray}
\end{lem}
The explicit construction of $x'$ and $x''$ from a given probability vector $x$ can be found in Refs. \cite{Brandao2015, Renner2004}.

\bibliographystyle{apsrev4-1}
 \bibliography{s-cat-coh-lit}

\end{document}